\documentclass[preliminary,copyright,creativecommons]{eptcs}
\providecommand{\event}{EXPRESS/SOS 2015} % Name of the event you are submitting to
\usepackage{breakurl}             % Not needed if you use pdflatex only.

\usepackage{amsmath}
\usepackage{amsthm}
\usepackage{thmtools}
\usepackage{mathtools}
\usepackage{enumitem}
\usepackage{relsize}
\usepackage{xspace}

\usepackage{nameref}
\usepackage{hyperref}
\usepackage[capitalise]{cleveref}

\declaretheorem[numberwithin=section,
	name=Theorem
	]{theorem}
\declaretheorem[sibling=theorem,
	name=Definition
	]{definition}
\declaretheorem[sibling=theorem,
	name=Lemma
	]{lemma}

\allowdisplaybreaks

\newcommand{\comb}[1]{\ensuremath{\mathsf{#1}}}

\DeclareMathOperator{\ter}{Ter}
\DeclareMathOperator{\lam}{\ensuremath{\Lambda}}
\newcommand{\sig}{\Sigma}
\newcommand{\sigSF}{\sig_{\textsf{SF}}}
\newcommand{\eqs}{\mathcal{E}}

\newcommand{\sfreduction}{\rightarrow_{\textsf{\relsize{-2}{SF}}}}

\newcommand{\CurryAppSF}{\textsf{SF}^{\mathcal{C}}_{\mathcal{@}}}
\newcommand{\casfreduction}{\rightarrow_{\textsf{\relsize{-2}{SF}}^{\mathcal{C}}_{@}}}

\newcommand{\fsymb}[1]{\mathop{\mathsf{#1}}}
\newcommand{\ct}[1]{\langle #1 \rangle}

\newcommand{\zero}{\fsymb{zero}}
\newcommand{\suc}{\fsymb{succ}}
\newcommand{\add}{\fsymb{add}}

\newcommand{\app}{\fsymb{app}}
\newcommand{\freduce}{\fsymb{f\textsf{-}reduce}}

\newcommand{\sem}[1]{[\![#1]\!]}

\newcommand{\Bohm}{B\"{o}hm\xspace}

\title{Encoding the Factorisation Calculus \\
  \smaller{(Representing the Intensional in the Extensional)}}
\author{Reuben N. S. Rowe
\institute{Department of Computer Science\\
University College London}
\email{r.rowe@ucl.ac.uk}
}
\def\titlerunning{Encoding the Factorisation Calculus}
\def\authorrunning{Reuben N. S. Rowe}
\begin{document}
\maketitle

\begin{abstract}
  Jay and Given-Wilson have recently introduced the Factorisation (or SF-) calculus as a minimal fundamental model of \emph{intensional} computation. It is a combinatory calculus containing a special combinator, \comb{F}, which is able to examine the internal structure of its first argument. The calculus is significant in that as well as being combinatorially complete it also exhibits the property of structural completeness, i.e.~it is able to represent any function on terms definable using pattern matching on arbitrary normal forms. In particular, it admits a term that can decide the structural equality of any two arbitrary normal forms.
  
  Since SF-calculus is combinatorially complete, it is clearly at least as powerful as the more familiar and paradigmatic Turing-powerful computational models of $\lambda$-calculus and Combinatory Logic. Its relationship to these models in the converse direction is less obvious, however. Jay and Given-Wilson have suggested that SF-calculus is strictly more powerful than the aforementioned models, but a detailed study of the connections between these models is yet to be undertaken.
  
  This paper begins to bridge that gap by presenting a faithful encoding of the Factorisation Calculus into the $\lambda$-calculus preserving both reduction and strong normalisation. The existence of such an encoding is a new result. It also suggests that there is, in some sense, an equivalence between the former model and the latter. We discuss to what extent our result constitutes an equivalence by considering it in the context of some previously defined frameworks for comparing computational power and expressiveness.
\end{abstract}

\section{Introduction}

Mathematical models of computation are useful in studying the formal properties of programming practice. Indeed, the field of computing today arose partly out of the study of such abstract models: namely Turing Machines \cite{Turing37a}, the $\lambda$-calculus \cite{Church36}, and Combinatory Logic \cite{Curry}, which are consequently considered to be archetypal computational models. 
It is standard practice to qualify the abilities, or expressiveness, of a formal model of computation by demonstrating that it may \emph{simulate} (and be simulated by) the operation of other formal models. This is the very essence of the notion of \emph{Turing-completeness}, which encapsulates the intuition that a model may carry out any operation that is `effectively computable'.
To construct such a simulation one must first give an injective mapping, showing how the terms of the source model may be represented by terms of the target. For example, this is the basis behind the process of G\"{o}delization and the Church encoding of natural numbers \cite{Kleene36}. Two basic properties are then required: that each atomic operational step of the source model is reflected by one or more steps of the target, and that a program of the target model \emph{terminates} whenever the corresponding source program does. The former is a key ingredient of Landin's influential work on comparing languages \cite{Landin66}, while the latter is used as a criterion for comparing expressiveness by, e.g., Felleisen \cite{Felleisen91}. Formal definitions of encodings incorporating these properties, referred to as ``faithful'', are already in use by the 90s, e.g.~in \cite{AbadiCV96}, and are now common (see e.g.~\cite{Gorla10}).

Recently, Jay's work on formal models of generic pattern matching \cite{JayBook} have led, in collaboration with Given-Wilson, to the formulation of the \emph{Factorisation Calculus}. This is a combinatory calculus comprising two combinators: the $\comb{S}$ combinator, familiar from Combinatory Logic; and a new $\comb{F}$ combinator. The purpose of the latter is to enable arbitrary (head) normal terms to be \emph{factorised}, that is split into their constituent parts, thereby allowing the examination of the internal structure of terms. This endows Factorisation Calculus with an interesting and powerful property: that of \emph{structure completeness}. This means that any function on terms themselves definable by pattern matching over \emph{arbitrary} normal forms is \emph{representable}. Thus, Factorisation Calculus can be viewed as a minimal, fundamental model characterising not only the abstract notion of \emph{pattern matching} but also \emph{intensional computation}. 

Jay and Given-Wilson show that Factorisation Calculus is Turing-complete, demonstrating a straightforward simulation of Combinatory Logic in their calculus. Moreover, due to its structural completeness, there is a term of Factorisation Calculus which can decide the structural equality of any two arbitrary normal forms. Conversely, factorisation and structural equality of normal forms \emph{cannot} be so represented in $\lambda$-calculus and Combinatory Logic, thus these models are \emph{not} structure complete. This hints at some sort of disparity in the expressivity of the two models. In their original and subsequent research \cite{Given-WilsonJ11,Given-Wilson14,JayV14}, Jay and Given-Wilson speculate that the added expressive power may manifest itself in a non-existence result for simulations of Factorisation Calculus in $\lambda$-calculus, but this is not pursued in detail.

We show that there \emph{does} exist a simulation of Factorisation Calculus within $\lambda$-calculus. The existence of such a simulation has not been demonstrated before, and this is the primary contribution of our paper. The simulation is made possible by a construction due to Berarducci and \Bohm, which shows how to encode a certain class of term rewriting systems within $\lambda$-calculus. We show that Factorisation Calculus can be simulated by such a term rewriting system, whence the result follows. In the classical framework, our result signifies that Factorisation Calculus is no more powerful than $\lambda$-calculus. Thus there appears to be a mismatch between our result and the structure completeness property that the standard simulation-based notion of equivalence does not account for. To begin to try and resolve this, we consider some research in the literature which refines the concept of computational equivalence and discuss how our result relates to this.

\paragraph{Outline} The rest of this paper is organised as follows. \cref{sec:Calculus} recalls Jay and Given-Wilson's Factorisation Calculus and its basic properties. \cref{sec:NormalSolutions} describes Berarducci and {\Bohm}'s construction for encoding so-called canonical rewrite systems within the $\lambda$-calculus. In \cref{sec:Encoding} we present our technical contribution: a simulation of the Factorisation Calculus in the $\lambda$-calculus via this construction. \cref{sec:Discussion} then discusses, in light of our results, how the relative expressiveness of Factorisation Calculus and $\lambda$-calculus may be characterised. \cref{sec:Conclusions} concludes and remarks on areas for future work.

\section{Factorisation Calculus}
\label{sec:Calculus}

We begin by presenting Jay and Given-Wilson's Factoriation Calculus itself, and review its principal properties. Factorisation Calculus, or more accurately SF-calculus\footnote{We may say that any combinatory calculus that is structure complete is \emph{a} factorisation calculus.}, is a combinatory calculus whose terms are those of the free algebra over the two-element signature containing the combinators $\comb{S}$ and $\comb{F}$, which each reduce upon being applied to three arguments. The former is the familiar combinator from Combinatory Logic \cite{Curry} which applies its first and second arguments to duplicates of its third. The $\comb{F}$ combinator, on the other hand, introduces new capabilities in the form of \emph{factorisation}: it is able to examine the \emph{internal} structure of its first argument and process its second and third in different ways depending on whether that argument is \emph{atomic} (i.e.~itself a combinator) or \emph{compound} (i.e.~a partial application). To illustrate this, consider how the $\comb{F}$ combinator reduces in the following two instances:
\begin{align*}
  \comb{F}\,\comb{S}\,M\,N &\rightarrow M
  &
  \comb{F}\,(\comb{S}\,X)\,M\,N &\rightarrow N\,\comb{S}\,X
\end{align*}
Observe that when the first argument is atomic, it eliminates its third argument and returns its second. On the other hand, when the first argument is compound it eliminates its second argument and \emph{factorises} the first into its left- and right-hand constituent components, passing these \emph{separately} to its third argument. 

Formally, SF-calculus is defined as follows.
\begin{definition}[{SF-calculus \cite[\textsection~4]{Given-WilsonJ11}}]
  The SF-calculus is a combinatory rewrite system over terms (ranged over by uppercase roman letters $M$, $N$, etc.) given by the following grammar:
  \begin{equation*}
    M,\,N \Coloneqq \comb{S} \;\mid\; \comb{F} \;\mid\; M\,N
  \end{equation*}
  Terms of the form $\comb{S}$, $\comb{F}$, $\comb{S}\,M$, $\comb{F}\,M$, $\comb{S}\,M\,N$, or $\comb{F}\,M\,N$ (i.e.~partially applied combinators) are called \emph{factorable forms}. Reduction of terms is the smallest contextually closed binary relation $\sfreduction$ on terms (with the reflexive transitive closure denoted by $\sfreduction^{\ast}$) satisfying:
  \begin{align*}
    \comb{S} \,M\,N\,X &\sfreduction M\,X\,(N\,X) \\
    \comb{F}\,O\,M\,N &\sfreduction M && \text{if $O$ is $\comb{S}$ or $\comb{F}$} \\
    \comb{F}\,(P\,Q)\,M\,N &\sfreduction N\,P\,Q && \text{if $P\,Q$ is a factorable form}
  \end{align*}
\end{definition}

Reduction of SF-calculus is confluent, and the $\comb{K}$ combinator of Combinatory Logic can be represented in SF-calculus by $\comb{F}\,\comb{F}$ (also, indeed, by $\comb{F}\,\comb{S}$). Thus, there is a trivial encoding of Combinatory Logic in SF-calculus which preserves reduction and strong normalisation \cite{Given-Wilson14}.

The behaviour of the $\comb{F}$ combinator gives SF-calculus an \emph{intensional} quality: one may define higher order functions in SF-calculus which discriminate between functions whose \emph{implementations} are different even when those functions are \emph{extensionally} equal (i.e.~produce identical outputs for identical inputs). For example, for any normal form $X$, the term $I_{X} \equiv \comb{S}\,(\comb{F}\,\comb{F})\,X$ implements the identity function (i.e.~$I_{X}\,M \sfreduction^{\ast} M$ for all $M$) and thus all such terms are extensionally equal. However, an SF-term $T$ can be constructed which distinguishes them (by behaving as $T\,I_{X} \sfreduction^{\ast} X$). Moreover, one can construct an SF-term that can decide the equality of any two \emph{arbitrary} normal forms.

The intensional behaviour of SF-calculus is formally characterised by a property called \emph{structure completeness}, which captures the notion that every \emph{symbolic computation} (i.e.~Turing-computable symbolic function) on normal forms is represented by some term.

\begin{definition}[{Structure Completeness \cite[\textsection~7-8]{Given-WilsonJ11}}]
  Let $\mathcal{C}$ be a confluent combinatory calculus whose terms include variables, with reduction relation $\rightarrow_{\mathcal{C}}^{\ast}$. Define \emph{patterns} to be the \emph{linear} normal forms (i.e.~containing no more than one occurrence of each variable), and \emph{matchable forms} to be partially applied combinators.
  \begin{enumerate}[nosep]
    \item A \emph{match} $\{ U / P \}$ of a pattern $P$ against a term $U$ may be defined to \emph{succeed} with a substitution of terms for variables, or \emph{fail} as follows (where $\mathsf{Id}$ denotes the identity function and $\uplus$ the disjoint union of substitutions with match failure as an absorbing element):
    \begin{gather*}
      \begin{gathered}
        \{ U / x \} = [U/x] 
        \qquad
        \{ A / A \} = \mathsf{Id} \\
        \{ U\,V / P\,Q \} = \{ U / P \} \uplus \{ V / Q \} \\
        \{ U / P \} = \mathsf{fail}
      \end{gathered}
      \qquad
      \begin{aligned}
        & \text{(if $A$ atomic)} \\
        & \text{(if\, $U\,V$ a compound)} \\
        & \text{(otherwise, if $U$ matchable)}
      \end{aligned}
    \end{gather*}
    \item A \emph{case} is an equation of the form $P = M$ where $P$ is a pattern and $M$ an arbitrary term, which defines a \emph{symbolic function} $\mathcal{G}$ on terms by $\mathcal{G}(U) = \sigma\,(M)$ if $\{ U / P \}$ succeeds with substitution $\sigma$, and $\mathcal{G}(U) = U$ if it fails.
    \item A confluent combinatory calculus is \emph{structure complete} if for every pattern $P$ and term $M$, there is some term $G$ such that $G\,U \rightarrow_{\mathcal{C}}^{\ast} \mathcal{G}(U)$ for every term $U$ on which $\mathcal{G}$ is defined; i.e.~the symbolic function defined by every case is represented by some term.
  \end{enumerate}
\end{definition}
\noindent
Structure completeness subsumes combinatorial completeness since $\lambda x . M$ is given by the case $x = M$.
\begin{theorem}[{\cite[Cor.~8.4]{Given-WilsonJ11}}]
\label{thm:StructuralCompleteness}
  SF-calculus is structure complete.
\end{theorem}

The $\comb{F}$ combinator itself represents a symbolic computation $\mathcal{F}$, namely that of factorisation:
\begin{align*}
  \mathcal{F}\,(A,\,M,\,N) &= M && \text{if $A$ is atomic} \\
  \mathcal{F}\,(P\,Q,\,M,\,N) &= N\,P\,Q && \text{if $P\,Q$ is compound}
\end{align*}
A significant (and arguably remarkable) fact is that $\mathcal{F}$ cannot be represented in Combinatory Logic (for definitions of atomic and compound appropriate thereto). %This result has a particularly beautiful proof.
\begin{theorem}[{\cite[Thm.~3.2]{Given-WilsonJ11}}]
\label{thm:NoFactorisationInCL}
  Factorisation of SK-combinators is a symbolic computation that is not representable in Combinatory Logic.
\end{theorem}
The equality predicate on normal forms also has no representation in Combinatory Logic. Thus, there exist (symbolic) functions, which are clearly `computable' from an empirical point of view, that are not (directly) representable in Combinatory Logic (there is also a similar result for $\lambda$-calculus \cite{Barendregt}). This result clearly points towards some form of added expressivity possessed by SF-calculus over the archetypal computational models. It is to this issue that we will return in \Cref{sec:Discussion}.

\section{Strongly Normalising Solutions of Equational Systems in $\lambda$-calculus}
\label{sec:NormalSolutions}

We now reiterate the interpretation result of Berarducci and \Bohm \cite{BerarducciB92}, upon which our technical contribution rests. Essentially, this result says that systems of equations for a particular class of term algebras can be given solutions in the $\lambda$-calculus such that the representation of each atomic term of the algebra is \emph{strongly normalising} thus having a \emph{normal form}. Moreover, when the set of equations is interpreted as a rewrite system the encoding of terms preserves reduction and strong normalisation.

We assume the usual definitions of the $\lambda$-calculus without further explanation (readers may refer to \cite{Barendregt} for details), with $\lam$ denoting the set of lambda terms, $\rightarrow_{\beta}^{\ast}$ denoting the (multi-step) $\beta$-reduction relation, and $=_{\beta}$ denoting $\beta$-equality (i.e.~the equivalence relation on lambda terms induced by $\beta$-reduction). Furthermore, we also assume the familiar algebraic notion of the set $\ter(\sig)$ of ($\sig$-)\emph{terms} over the signature (set of \emph{function symbols}, each with an associated arity) $\sig$. We can then also consider the set $\lam(\sig)$ of \emph{extended} lambda terms (i.e.~lambda terms which may contain $\sig$-terms); notice that both $\ter(\sig)$ and $\lam$ are (strict) subsets of $\lam(\sig)$.

\begin{definition}[Canonical Systems of Equations]
\label{def:CanonicalEquations}
  Fix a signature $\sig$ and let $\eqs$ be a set of equations between terms $t \in \ter(\sig)$. We say that $\eqs$ is \emph{canonical} if $\sig$ can be partitioned into two disjoint subsets $\sig_0$ and $\sig_1$ (i.e.~$\sig = \sig_0 \cup \sig_1$) such that: each equation in $\eqs$ is of the form $\mathop{f} (\mathop{c} (x_1, \ldots, x_m), y_1, \ldots, y_n) = t$ with $c \in \sig_0$ and $f \in \sig_1$ and where the variables $x_1, \ldots, x_m, y_1, \ldots, y_n$ are all distinct and form a superset of the variables in the term $t$; and for each distinct pair $(c, f) \in \sig_0 \times \sig_1$ there is at most one equation in $\eqs$ of this form. We say that $\eqs$ is \emph{complete} if for each distinct pair $(c, f) \in \sig_0 \times \sig_1$ there is exactly one such equation in $\eqs$.
\end{definition}

Canonical systems of equations, then, partition the signature into a set $\sig_0$ of (algebraic datatype) \emph{constructors}, and $\sig_1$ of \emph{programs} defined by pattern matching over the constructors on the first argument. Notice that any incomplete canonical system of equations can trivially be made complete by adding equations for the missing cases which simply project one of the function's arguments\footnote{Alternatively, one might want to introduce a new nullary constructor (denoting an `error' value) and add equations for the missing cases that simply return this value.}.

As an example of a canonical system of equations, we may observe that the usual recursive definition of addition over the datatype of (Peano) natural numbers is such a system:
  \begin{equation*}
    \add \, (\zero , x) = x
    \qquad
    \add \, (\suc \, (x) , y) = \suc \, (\add \, (x, y))
  \end{equation*}
We have a signature containing one function symbol $\mathsf{add}$, one nullary constructor $\mathsf{zero}$, and one unary constructor $\mathsf{succ}$; moreover in this simple case, the equation system is already complete. In fact, every partial recursive function (on natural numbers) can be defined by a canonical system of equations \cite{BerarducciB92,BohmPG94}.

Given an equational system $\eqs$ over a signature $\sig$, we can also take it to define a term rewriting system on $\ter(\sig)$ by reading each equation as a \emph{rewrite} rule, i.e.~$\mathop{f_i} (\mathop{c_j} (x_1, \ldots, x_m), y_1, \ldots, y_n) \rightarrow t$. We will write $\rightarrow_{\eqs}$ for the (one-step) reduction relation of the rewrite system defined by $\eqs$ in this way (i.e.~the smallest binary relation on terms satisfying the rewrite rules and closed under substitution and contexts), and $\rightarrow^{\ast}_{\eqs}$ for its reflexive, transitive closure (i.e.~multi-step reduction).
Ultimately, the aim is to interpret equational systems (and their associated rewrite systems) within $\lambda$-calculus.

\begin{definition}[Interpretations]
  A \emph{representation} of the signature $\sig$ is a function $\phi : {\sig \rightarrow \Lambda}$ from the function symbols of $\sig$ to (closed) lambda terms, and induces a map $(\cdot)^{\phi} : {\lam(\sig) \rightarrow \Lambda}$ in the obvious way, namely by $x^{\phi} = x$, $(\lambda x . M)^{\phi} = \lambda x . M^{\phi}$, $(MN)^{\phi} = {M^{\phi}}{N^{\phi}}$, and for $f \in \sig$, ${f \, (t_1, \ldots, t_n)}^{\phi} = {\phi(f) \, {t_1}^{\phi} \, \ldots \, {t_n}^{\phi}}$. We say that a representation $\phi$ \emph{satisfies} (or \emph{solves}) $\eqs$ if for each equation $t_1 = t_2$ (and corresponding rewrite rule $t_1 \rightarrow t_2$) in $\eqs$ we have ${t_1}^{\phi} =_{\beta} {t_2}^{\phi}$ (and correspondingly also ${t_1}^{\phi} \rightarrow_{\beta}^{\ast} {t_2}^{\phi}$). When a representation $\phi$ satisfies $\eqs$, we say that $\phi$ is an \emph{interpretation} (or a \emph{solution}) of $\eqs$ within $\lambda$-calculus.
\end{definition}

The following construction gives a special kind of representation for canonical systems of equations.

\begin{definition}[Canonical Representations]
\label{def:CanonicalRepresentation}
  Let $\eqs$ be a canonical system of equations that partitions the signature $\sig$ into constructors $\sig_0 = \{c_1, \ldots, c_r\}$ and programs $\sig_1 = \{f_1, \ldots, f_k\}$. Without loss of generality we may assume that $\eqs$ is complete, and so for each $1 \leq i \leq k$ and $1 \leq j \leq r$ let $b_{(i,j)}$ denote the term $t$ such that $\mathop{f_i} (\mathop{c_j} (x_1, \ldots, x_m), y_1, \ldots, y_n) = t \in \eqs$. 
  
  We will make use of the following notational abbreviations:
  \begin{itemize}[nosep,label=-,labelindent=\parindent, leftmargin=2\parindent]
    \item 
    Let $\ct{t_1, \ldots, t_n}$ denote the Church n-tuple, i.e.~$\lambda x . x t_1 \ldots t_n$.
    \item 
    Let $\Pi^{n}_{k}$ (where $1 \leq k \leq n$) be the $n$-ary $k$\textsuperscript{th} projection function, i.e.~$\lambda x_1 \ldots x_n . x_k$.
    \item 
    For $k \geq i > 1$, let $t_i, \ldots, t_k, \ldots, t_{i-1}$ denote the cyclic permutation of $t_1, t_2, \ldots, t_k$ beginning with $t_i$; (in an abuse of notation we may also take $t_i, \ldots, t_k, \ldots, t_{i-1} = t_1, t_2, \ldots, t_k$ when $i = 1$).
  \end{itemize}
  We now define two disjoint representations $\vartheta$ and $\zeta$ for constructors and programs respectively.
  \begin{description}[nosep,labelindent=\parindent, leftmargin=2\parindent, format=\normalfont]
    \item[(Representation of Constructors)] 
    For each $1 \leq i \leq r$, we define the representation of the constructor $c_i$ as follows (where $n$ is the arity of $c_i$):
    \begin{equation*}
      \vartheta(c_i) = \lambda x_1 \ldots x_n f . f \Pi^{r}_{i} x_1 \ldots x_n f
    \end{equation*}
    \item[(Representation of Programs)] 
    We choose $k$ distinct fresh variables $v_1, \ldots, v_k$ not occurring in $\eqs$ and fix a `pre-representation', $\psi$, of $\sig_1$ defined by $\psi(f_i) = \ct{v_i, \ldots, v_k, \ldots, v_{i-1}}$ for each $1 \leq i \leq k$.
    Using this representation, and the representation of constructors defined above, we then define $k \times r$ lambda terms $t_{(i,j)}$ ($1 \leq i \leq k$, $1 \leq j \leq r$), using the equations in $\eqs$ as follows:
    \begin{equation*}
      t_{(i,j)} = \lambda x_1 \ldots x_m v_i \ldots v_k \ldots v_{i-1} y_1 \ldots y_n . ({b_{(i,j)}}^{\psi})^{\vartheta}
    \end{equation*}
    where $\mathop{f_i} (\mathop{c_j} (x_1, \ldots, x_m), y_1, \ldots, y_n) = b_{(i,j)} \in \eqs$ is the equation defining the behaviour of $f_i$ when given a datum constructed using $c_j$ as its first argument.
    We now define $k$ terms, each one a Church $r$-tuple collating the bodies of all the cases for one of the programs in $\sig_1$, as follows:
    \begin{equation*}
      t_i = \ct{t_{(i,1)}, \ldots, t_{(i,r)}}
    \end{equation*}
    (where $1 \leq i \leq k$).
    Each program is then represented by a Church $k$-tuple containing the collated representations of each program defintion, beginning with its own. That is, $\zeta$ is defined by:
    \begin{equation*}
      \zeta(f_i) = \ct{t_i, \ldots, t_k, \ldots, t_{i-1}} \quad \text{($1 \leq i \leq k$)}
    \end{equation*}
  \end{description}
  The representation $\phi = \vartheta \cup \zeta$ is called a \emph{canonical} representation of $\sig$ with respect to $\eqs$.
\end{definition}

To gain some insight into the construction defined above, one can observe that it is related to an encoding of data attributed to Scott\footnote{The citation can be found in Curry, Hindley and Seldin \cite[p. 504]{Curry2}.} (and thus commonly referred to in the literature as the Scott encoding), which has subsequently been developed by others (e.g.~\cite{Steensgaard-Madsen89,Mogensen92,JansenKP06,Stump09}). In the more familiar `standard' encoding of functions, a fixed-point combinator is used to solve any recursion in the definition. This has the effect of making recursion \emph{explicit}, and thus the representations of recursive functions have infinite expansions consisting of a `list' of distinct instances of the function body, one for each recursive call that may be made. Applying the function to a datum then corresponds to a \emph{fold} of the datum over this list, which discards the remaining infinity of recursive calls once the base case is reached. Therefore, as described by \Bohm et al. \cite[\textsection 3]{BohmPG94}, in this scheme functions are `diverging objects which, when applied to data, may ``incidentally'' converge'. In encodings of the Scott variety, the recursive nature of functions is kept \emph{implicit} and, while still triggered by application to a datum, only reproduced `on demand'. Hence we obtain finite objects which now `may ``incidentally'' diverge' when applied to data\footnote{This reversed form of the slogan is also due to \Bohm et al., and illustrates the dual nature of the Scott and Church encodings.}.

To explicate the particular encoding specified by \Cref{def:CanonicalRepresentation}, we point out that the representation of a constructor is a (lambda) function that takes in the appropriate number of arguments (the \emph{sub-data} of the datum that is subsequently constructed) and then waits to be given a function, which will be the program to be executed. Now, looking at how the constructor representation uses this function argument, we see that programs should expect to be given a projection function, followed by a number of sub-data, and then they are also passed \emph{a copy of themselves}. It is this final element which is the key to Scott-type encodings, and allows recursion to be kept implicit. Looking now at the representation of programs we see that they are Church $k$-tuples containing an element for each program defined by $\eqs$ (each of which is a Church $r$-tuple, where each element is a representation of one of the cases of that program's definition). Thus the representation of each program contains the definition of \emph{every} program defined by $\eqs$; in particular it will contain the definition of each program which it may itself invoke. To illustrate in more detail how the encoding works, we can consider the general reduction sequence of a term representing the application of some program $\comb{prog}_i$ to some arguments, the first of which is a datum constructed as $c_j\,(d_1, \ldots, d_m)$:
  \begin{align}
    \mathrlap{{(\comb{prog}_i \, (c_j \, (d_1, \ldots, d_m)) \, \comb{arg}_1 \, \ldots \, \comb{arg}_m )}^{\phi}
      = \ct{t_i, \ldots, t_{i-1}} \, {(c_j \, (d_1, \ldots, d_m))}^{\phi} \, \comb{arg}_1^{\phi} \, \ldots \, \comb{arg}_m^{\phi}
      } \quad& \notag \\
      &\rightarrow_{\beta}^{\ast} &&(\lambda x . x \, t_i \ldots t_{i-1}) \, (\lambda f . f \Pi^{r}_{j} d_1 \ldots d_m f) \, \comb{arg}_1^{\phi} \, \ldots \, \comb{arg}_m^{\phi} \label{eqn:ApplyProg} \\
      &\rightarrow_{\beta} &&(\lambda f . f \Pi^{r}_{j} d_1 \ldots d_m f) \, t_i \ldots t_{i-1} \, \comb{arg}_1^{\phi} \, \ldots \, \comb{arg}_m^{\phi} \label{eqn:PassProgs} \\
      &\rightarrow_{\beta} &&t_i \, \Pi^{r}_{j} d_1 \ldots d_m \, t_i \, t_{i+1} \ldots t_{i-1} \, \comb{arg}_1^{\phi} \, \ldots \, \comb{arg}_m^{\phi} \label{eqn:CaseSelection:Start} \\
      &= &&\ct{t_{i,1}, \ldots, t_{i,r}} \, \Pi^{r}_{j} d_1 \ldots d_m \, t_i \, t_{i+1} \ldots t_{i-1} \, \comb{arg}_1^{\phi} \, \ldots \, \comb{arg}_m^{\phi} \\
      &= &&(\lambda x . x \, t_{i,1} \ldots t_{i,r}) \, \Pi^{r}_{j} d_1 \ldots d_m \, t_i \, t_{i+1} \ldots t_{i-1} \, \comb{arg}_1^{\phi} \, \ldots \, \comb{arg}_m^{\phi} \\
      &\rightarrow_{\beta} &&\Pi^{r}_{j}  \, t_{i,1} \ldots t_{i,r} \, d_1 \ldots d_m \, t_i \, t_{i+1} \ldots t_{i-1} \, \comb{arg}_1^{\phi} \, \ldots \, \comb{arg}_m^{\phi} \\
      &\rightarrow_{\beta} &&t_{i,j} \, d_1 \ldots d_m \, t_i \, t_{i+1} \ldots t_{i-1} \, \comb{arg}_1^{\phi} \, \ldots \, \comb{arg}_m^{\phi} \label{eqn:CaseSelection:End} \\
      &= &&(\lambda x_1 \ldots x_m v_i \ldots v_k \ldots v_{i-1} y_1 \ldots y_n . ({b_{(i,j)}}^{\psi})^{\vartheta}) \, d_1 \ldots d_m \, t_i \, t_{i+1} \ldots t_{i-1} \, \comb{arg}_1^{\phi} \, \ldots \, \comb{arg}_m^{\phi} \notag \\
      &\rightarrow_{\beta}^{\ast} &&{(b_{(i,j)} \, [d_1 / x_1, \ldots, d_m / x_m, \comb{arg}_1 / y_1, \ldots, \comb{arg}_n / y_n])}^{\phi} \label{eqn:SubstFunBody}
  \end{align}
When the program is applied to a datum (\cref{eqn:ApplyProg}), its representation arranges to apply the datum first to the representations of each program beginning with its own, and then to the remainder of the arguments (\cref{eqn:PassProgs}). Then, the particular structure of the datum will reduce the expression to pick out the appropriate case of the program definition to be executed, and apply it to the sub-data and the representations of each program, having duplicated the program being executed (\crefrange{eqn:CaseSelection:Start}{eqn:CaseSelection:End}). This then reduces to the representation of the appropriate substitution instance of the function body (\cref{eqn:SubstFunBody}).

The result of Berarducci and \Bohm says that a canonical representation gives an interpretation that also preserves strong normalisation.
\begin{theorem}[Interpretation Theorem {\cite[Thm 3.4]{BerarducciB92}}]
\label{thm:Interpretation}
  Let $\sig$ be a signature and $\eqs$ a canonical set of equations for $\sig$; then any canonical representation $\phi$ for $\sig$ with respect to $\eqs$ is an interpretation of $\eqs$ within $\lambda$-calculus. In addition $(\cdot)^{\phi}$ preserves strong normalisation of closed terms.
\end{theorem}

\section{Encoding the Factorisation Calculus}
\label{sec:Encoding}

In this section, we present our novel technical contribution: an encoding of SF-calculus in $\lambda$-calculus. We believe that this is the first such encoding presented in the literature. Our encoding is a faithful simulation; it preserves both the reduction behaviour of terms (thus also $\beta$-equality) and their termination behaviour (i.e.~strong normalisation). In this section we shall make use of standard notation and results for term rewriting systems, details of which may be found in \cite{Klop-TRS}.

The key step to the encoding is to define a rewrite system behaviourally equivalent to SF-calculus that is also canonical, in the sense of \Cref{def:CanonicalEquations}. It is then simply a matter of applying the construction of Berarducci and \Bohm to obtain the encoding. Thus it is our translation of SF-calculus into this intermediate rewrite system that is the primary novelty of our contribution.

We are aiming to derive a set of rewrite rules that is \emph{canonical} and so we must translate the schematic definition of Jay and Given-Wilson, as presented in \Cref{sec:Calculus}, into one consisting of algebraic rewrite rules. There are two salient features of \Cref{def:CanonicalEquations} that we must take into account: that the rewrite rules must make a distinction between \emph{programs} and \emph{constructors}; and that the left-hand side of each rewrite rule must contain exactly one program symbol and one constructor. To obtain rewrite rules of the required form, we recast SF-calculus as a \emph{curryfied}, \emph{applicative} term rewriting system. That is, we first introduce an explicit program symbol $\app$ to denote application and use the symbols $\comb{S}$ and $\comb{F}$ solely as \emph{constructors}. Secondly we stratify the combinators $\comb{C} \in \{\comb{S}, \comb{F}\}$ into sets $\{\comb{C}_0, \comb{C}_1, \comb{C}_2\}$ of constructors, each of which represent successive \emph{partial} applications of their underlying combinator $\comb{C}$. Although this `currying' process is well-known from the world of functional programming, the reader may refer to \cite{Kahrs95,Bakel-Fernandez-IaC'96} for a formal definition of this process in the context of general term rewriting.

We may take the rewrite rules for the $\comb{S}$ combinator directly from the standard curryfied applicative formulation of Combinatory Logic (see e.g.~\cite{Bakel-Fernandez-IaC'96}):
\begin{align*}
  \app\,(\comb{S_0},\,x) & \rightarrow \comb{S_1}\,(x)
  &
  \app\,(\comb{S_1}\,(x),\,y) &\rightarrow \comb{S_2}\,(x,\,y)
  &
  \app\,(\comb{S_2}\,(x,\,y),\,z) &\rightarrow \app\,(\app\,(x,\,z),\,\app\,(y,\,z))
\end{align*}
The rules for producing the partial applications of the $\comb{F}$ combinator are similarly straightforward:
\begin{align*}
  \app\,(\comb{F_0},\,x) & \rightarrow \comb{F_1}\,(x)
  &
  \app\,(\comb{F_1}\,(x),\,y) &\rightarrow \comb{F_2}\,(x,\,y)
\end{align*}
The rewrite rule for the full application of the $\comb{F}$ combinator is more tricky because we must find a way of implementing its two possible reductions. As in the original formulation of SF-calculus, since the choice of which reduction to make is determined by the structure of the first argument we should like to be able to use the pattern-matching capabilities inherent in the term rewriting discipline, e.g. by giving rewrite rules such as:
\begin{align*}
  \app\,(\comb{F_2}\,(\comb{S_0},\,y),\,z) &\rightarrow y
  &
  \app\,(\comb{F_2}\,(\comb{F_1}\,(x),\,y),\,z) &\rightarrow \app\,(\app\,(z,\,\comb{F_0}),\,x)
\end{align*}
However these rules are \emph{not} canonical since they contain two occurrences of a constructor: they are pattern-matching `too deeply'. We can circumvent this by introducing an auxiliary \emph{program} symbol $\freduce$ and then having the rewrite rule for the $\comb{F_2}$ case of $\app$ delegate to this new program:
\begin{equation*}
  \app\,(\comb{F_2}\,(x,\,y),\,z) \rightarrow \freduce\,(x,\,y,\,z)
\end{equation*}
Since $\freduce$ is an independent program symbol, and only needs to pattern match on its first argument to determine which result to compute, we may give canonical rewrite rules for it, such as the following:
\begin{align*}
  \freduce\,(\comb{S_0},\,y,\,z) &\rightarrow y
  &
  \freduce\,(\comb{F_1}\,(x),\,y,\,z) \rightarrow \app\,(\app\,(z,\,\comb{F_0}),\,x)
\end{align*}

We now have all the components to be able to present a canonical rewrite system that faithfully implements SF-calculus.

\begin{definition}[Currified Applicative SF-Calculus]
\label{def:CompleteTRS}
  Let $\sigSF = \sig_0 \cup \sig_1$ be the signature comprising the set $\sig_0 = \{ \comb{S_0}, \comb{S_1}, \comb{S_2}, \comb{F_0}, \comb{F_1}, \comb{F_2} \}$ of constructors and the set $\sig_1 = \{ \app, \freduce \}$ of programs. \emph{Curryfied Applicative SF-calculus} is the term rewriting system $\CurryAppSF$ defined by the rewrite rules given in \Cref{fig:ApplicativeCalculus} over the signature $\sigSF$. We denote its one-step and many-step reduction relations by $\casfreduction$ and $\casfreduction^{\ast}$.% respectively.
\end{definition}
\noindent
Notice that the rewrite rules of $\CurryAppSF$ are canonical, in the sense of \Cref{def:CanonicalEquations}, and that they are also \emph{complete}. We also remark that $\CurryAppSF$ is an \emph{orthogonal} term rewriting system \cite[Def. 2.1.1]{Klop-TRS}.

\begin{figure}[t]
  \begin{gather*}
    \begin{aligned}
      &
      \app\,(\comb{S_0},\,x) \rightarrow \comb{S_1}\,(x)
      & \qquad &
      \app\,(\comb{F_0},\,x) \rightarrow \comb{F_1}\,(x)
      \\
      &
      \app\,(\comb{S_1}\,(x),\,y) \rightarrow \comb{S_2}\,(x,\,y)
      &&
      \app\,(\comb{F_1}\,(x),\,y) \rightarrow \comb{F_2}\,(x,\,y)
      \\
      &
      \app\,(\comb{S_2}\,(x,\,y),\,z) \rightarrow \app\,(\app\,(x,\,z),\,\app\,(y,\,z))
      &&
      \app\,(\comb{F_2}\,(x,\,y),\,z) \rightarrow \freduce\,(x,\,y,\,z)
    \end{aligned}
    \\[0.5em]
    \begin{aligned}
      \freduce\,(\comb{S_0},\,y,\,z) & \rightarrow y
      &\qquad 
      \freduce\,(\comb{S_1}\,(x),\,y,\,z) & \rightarrow \app\,(\app\,(z,\,\comb{S_0}),\,x)
      \\
      \freduce\,(\comb{F_0},\,y,\,z) & \rightarrow y
      &
      \freduce\,(\comb{F_1}\,(x),\,y,\,z) & \rightarrow \app\,(\app\,(z,\,\comb{F_0}),\,x)
    \end{aligned}
    \\[0.5em]
    \begin{aligned}
      \freduce\,(\comb{S_2}\,(p,\,q),\,y,\,z) & \rightarrow \app\,(\app\,(z,\,\app\,(\comb{S_0},\,p)),\,q)
      \\
      \freduce\,(\comb{F_2}\,(p,\,q),\,y,\,z) & \rightarrow \app\,(\app\,(z,\,\app\,(\comb{F_0},\,p)),\,q)
    \end{aligned}
  \end{gather*}
  \caption{A Complete Set of Canonical Rewrite Rules for Currified Applicative SF-calculus}
  \label{fig:ApplicativeCalculus}
\end{figure}

It is interesting to observe that our implementation of SF-calculus as a canonical applicative term rewriting system has involved an application of the \emph{Visitor} design pattern\footnote{In fact, the encoding that we are presenting in this paper arose as a direct result of considering how the Factorisation Calculus could be implemented in (Featherweight) Java.} \cite{GammaEtAl95}. When it comes to reducing a complete application of the $\comb{F}$ combinator, the computation must proceed based on the particular identity of some object (i.e.~the first argument), but \emph{without having any knowledge of that identity}. The solution is to apply the visitor pattern, which involves invoking a new `\textsf{visit}' operation (that we call $\freduce$) on the object, which in response executes the appropriate behaviour based on its \emph{self}-knowledge of its own identity. We do not think it is entirely coincidental that the visitor pattern has arisen in our work: its connection with structural matching has already been noted \cite{PalsbergJ98}, and investigating this connection further is an avenue for future research.

There is a straightforward translation from SF-calculus to $\CurryAppSF$.
\begin{definition}[Translation of SF-calculus to $\CurryAppSF$]
  The translation $\sem{\cdot}_{@}$ from SF-terms to $\CurryAppSF$-terms is defined by $\sem{\comb{S}}_{@} = \comb{S_0}$, $\sem{\comb{F}}_{@} = \comb{F_0}$, and $\sem{MN}_{@} = \app\,(\sem{M}_{@},\,\sem{N}_{@})$.
\end{definition}

We now show that $\CurryAppSF$ faithfully implements SF-calculus.

\begin{lemma}[$\sem{\cdot}_{@}$ Preserves Reduction]
\label{lem:Preservation:Reduction}
  Let $M$ and $N$ be SF-terms; if $M \sfreduction^{\ast} N$ then $\sem{M}_{@} \casfreduction^{\ast} \sem{N}_{@}$.
\end{lemma}
\begin{proof}
  It is sufficient to consider the basic reduction rules of SF-calculus. In the interests of clarity, we underline the redex that is contracted at each step. The case for $\comb{S}$ is straightforward:
  \begin{align*}
    \sem{\comb{S}\,M\,N\,X}_{@} 
    & = \app\,(\app\,(\underline{\app\,(\comb{S_0},\,\sem{M}_{@})},\,\sem{N}_{@}),\,\sem{X}_{@}) 
      \casfreduction {\app\,(\underline{\app\,(\comb{S_1}\,(\sem{M}_{@}),\,\sem{N}_{@})},\,\sem{X}_{@})} \\
    & \casfreduction \underline{\app\,(\comb{S_2}\,(\sem{M}_{@},\,\sem{N}_{@}),\,\sem{X}_{@})}
      \casfreduction {\app\,(\app\,(\sem{M}_{@},\,\sem{X}_{@}),\,\app\,(\sem{N}_{@},\,\sem{X}_{@}))} \\
    & = \sem{M\,X\,(N\,X)}_{@}
  \end{align*}
  The case for $\comb{F}$ with $\comb{S}$ the first argument (i.e.~atomic) is as follows (the other atomic case is symmetric):
  \begin{multline*}
    \sem{\comb{F}\,\comb{S}\,M\,N}_{@} 
    = {\app\,(\app\,(\underline{\app\,(\comb{F_0},\,\comb{S_0})},\,\sem{M}_{@}),\,\sem{N}_{@})} 
    \casfreduction {\app\,(\underline{\app\,(\comb{F_1}\,(\comb{S_0}),\,\sem{M}_{@})},\,\sem{N}_{@})} \\
    \casfreduction \underline{\app\,(\comb{F_2}\,(\comb{S_0},\,\sem{M}_{@}),\,\sem{N}_{@})} 
    \casfreduction \underline{\freduce\,(\comb{S_0},\,\sem{M}_{@},\,\sem{N}_{@})}
    \casfreduction \sem{M}_{@}
  \end{multline*}
  When the first argument to $\comb{F}$ is a factorable form, we must further consider its structure. The case for when the first argument is $\comb{S}\,X$ (for some term $X$) is as follows:
  \begin{align*}
    \mathrlap{\sem{\comb{F}\,(\comb{S}\,X)\,M\,N}_{@} 
    = {\app\,(\app\,(\app\,(\comb{F_0},\,\underline{\app\,(\comb{S_0},\,\sem{X}_{@})}),\,\sem{M}_{@}),\,\sem{N}_{@})}} \quad & \\
    &\casfreduction {\app\,(\app\,(\underline{\app\,(\comb{F_0},\,\comb{S_1}\,(\sem{X}_{@}))},\,\sem{M}_{@}),\,\sem{N}_{@})} \\
    &\casfreduction {\app\,(\underline{\app\,(\comb{F_1}\,(\comb{S_1}\,(\sem{X}_{@})),\,\sem{M}_{@})},\,\sem{N}_{@})} 
    \casfreduction \underline{\app\,(\comb{F_2}\,(\comb{S_1}\,(\sem{X}_{@}),\,\sem{M}_{@}),\,\sem{N}_{@})} \\
    &\casfreduction \underline{\freduce\,(\comb{S_1}\,(\sem{X}_{@}),\,\sem{M}_{@},\,\sem{N}_{@})}
    \casfreduction {\app\,(\app\,(\sem{N}_{@},\,\comb{S_0}),\,\sem{X}_{@})} 
    = \sem{N\,\comb{S}\,X}_{@}
  \end{align*}
  Again, the case for when the first argument is $\comb{F}\,X$ (for some term $X$) is symmetric and can be obtained from the above sequence by replacing each occurrence of $\comb{S_0}$ by $\comb{F_0}$ and each occurrence of $\comb{S_1}$ by $\comb{F_1}$. 
  
  The case for when the first argument is $\comb{S}\,X\,Y$ (for some terms $X$ and $Y$) is as follows:
  \begin{align*}
    \mathrlap{\sem{\comb{F}\,(\comb{S}\,X\,Y)\,M\,N}_{@} 
    = {\app\,(\app\,(\app\,(\comb{F_0},\,\app\,(\underline{\app\,(\comb{S_0},\,\sem{X}_{@})},\,\sem{Y}_{@})),\,\sem{M}_{@}),\,\sem{N}_{@})}} \quad & \\
    &\casfreduction {\app\,(\app\,(\app\,(\comb{F_0},\,\underline{\app\,(\comb{S_1}\,(\sem{X}_{@}),\,\sem{Y}_{@})}),\,\sem{M}_{@}),\,\sem{N}_{@})} \quad & \\
    &\casfreduction {\app\,(\app\,(\underline{\app\,(\comb{F_0},\,\comb{S_2}\,(\sem{X}_{@},\,\sem{Y}_{@}))},\,\sem{M}_{@}),\,\sem{N}_{@})} \\
    &\casfreduction {\app\,(\underline{\app\,(\comb{F_1}\,(\comb{S_2}\,(\sem{X}_{@},\,\sem{Y}_{@})),\,\sem{M}_{@})},\,\sem{N}_{@})} \\
    &\casfreduction \underline{\app\,(\comb{F_2}\,(\comb{S_2}\,(\sem{X}_{@},\,\sem{Y}_{@}),\,\sem{M}_{@}),\,\sem{N}_{@})} \\
    &\casfreduction \underline{\freduce\,(\comb{S_2}\,(\sem{X}_{@},\,\sem{Y}_{@}),\,\sem{M}_{@},\,\sem{N}_{@})} \\
    &\casfreduction {\app\,(\app\,(\sem{N}_{@},\,\app\,(\comb{S_0},\,\sem{X}_{@})),\,\sem{Y}_{@})} 
    = \sem{N\,(\comb{S}\,X)\,Y}_{@}
  \end{align*}
  Once more, the case for when the first argument is $\comb{F}\,X\,Y$ (for some terms $X$ and $Y$) is symmetric and can be obtained from the above sequence by replacing each occurrence of $\comb{S_0}$ by $\comb{F_0}$, each occurrence of $\comb{S_1}$ by $\comb{F_1}$, and each occurrence of $\comb{S_2}$ by $\comb{F_2}$. %This concludes the proof.
\end{proof}

To show that $\sem{\cdot}_{@}$ preserves strong normalisation, we will rely on the notion of a \emph{perpetual reduction sequence}. We recall the relevant definitions of perpetual reductions and their properties \cite{KhasidashviliOO01}. A term $t$ is called an \emph{$\infty$-term} (also denoted $\infty(t)$) if it has an infinite reduction sequence. A reduction \emph{step} $t \rightarrow s$ is called \emph{perpetual} if $\infty(t)$ implies $\infty(s)$, that is it preserves divergence, and a reduction sequence $t_1 \rightarrow \ldots \rightarrow t_n$ is perpetual if every step $t_i \rightarrow t_{i+1}$ is perpetual. Clearly, a perpetual reduction sequence $t \rightarrow^{\ast} t'$ also preserves divergence. %, i.e.~$\infty(t)$ implies $\infty(t')$. 
A \emph{redex} $u$ is called perpetual if its contraction in every context yields a perpetual reduction step. Let $u \rightarrow t$ be a substitution instance of a rewrite rule $r$ (so $u$ is a redex and $t$ its $r$-contraction), then call the subterms of $u$ that are those substituted for the variables in $r$ the \emph{arguments} of $u$. Such an argument is said to be \emph{erased} if it corresponds to a variable that does \emph{not} occur in the right-hand side of $r$. It is the case that for orthogonal rewrite systems every redex whose erased arguments are strongly normalising and closed (i.e.~containing no variables) is perpetual \cite[Cor.~5.1]{KhasidashviliOO01}.

We first prove a couple of auxiliary lemmas.
\begin{lemma}
\label{lem:SF-NFImagesSN}
  Let $N$ be an SF-normal form, then $\sem{N}_{@}$ is strongly normalising.
\end{lemma}
\begin{proof}
  We characterise the normal forms as terms taking one of the following forms: $\comb{S}$, $\comb{F}$, $\comb{S}\,X$, $\comb{F}\,X$, $\comb{S}\,X\,Y$ or $\comb{F}\,X\,Y$, in which each subterm is also a normal form. We then proceed by induction on the size of terms. The base cases, i.e.~when $N$ is either $\comb{S}$ of $\comb{F}$ are trivial since then $\sem{N}_{@}$ is itself a normal form. For the inductive cases, notice that the terms $\comb{S_1}\,(\sem{X}_{@})$, $\comb{F_1}\,(\sem{X}_{@})$, $\comb{S_2}\,(\sem{X}_{@},\,\sem{Y}_{@})$ and $\comb{F_2}\,(\sem{X}_{@},\,\sem{Y}_{@})$ are strongly normalising since they are head normal and by induction $\sem{X}_{@}$ and $\sem{Y}_{@}$ are strongly normalising as $X$ and $Y$ are by definition normal forms (smaller than $N$). It is then straightforward to show in each case that the (unique) reduction from $\sem{N}_{@}$ to its corresponding head normal form given above is perpetual since it does not erase any arguments, and $\CurryAppSF$ is an orthogonal rewrite system. The result then follows.
\end{proof}

\begin{lemma}
\label{lem:PerpetualReductions}
  Let $X$, $M$ and $N$ be SF-normal forms, and $O$ an operator (i.e.~either $\comb{S}$ or $\comb{F}$) such that $O\,X\,M\,N \sfreduction R$; then $\mathcal{C}[\sem{O\,X\,M\,N}_{@}] \casfreduction^{\ast} \mathcal{C}[\sem{R}_{@}]$ is a perpetual reduction sequence for any $\CurryAppSF$-term context $\mathcal{C}$.
\end{lemma}
\begin{proof}
  We show that there is a reduction sequence $\sem{O\,X\,M\,N}_{@} \casfreduction^{\ast} \sem{R}_{@}$ which contracts a perpetual redex at each step, from which the result immediately follows. In fact, the reduction sequences that witness this are exactly those that are used to show preservation of reduction. In each case notice that, in the reduction sequence demonstrated in the proof of \Cref{lem:Preservation:Reduction}, the only erased argument (when it exists) is in the final reduction step and in each case this argument is either $\sem{M}_{@}$ or $\sem{N}_{@}$ which by \Cref{lem:SF-NFImagesSN} is strongly normalising since $M$ and $N$ are normal forms (and also closed since we do not consider SF-terms with variables). Thus, since $\CurryAppSF$ is an orthogonal rewrite system, it follows that the redex contracted at each step is perpetual.
\end{proof}

We can now prove the following result.

\begin{lemma}[$\sem{\cdot}_{@}$ Preserves Strong Normalisation]
\label{lem:Preservation:StrongNormalisation}
  Let $M$ be a strongly normalising SF-term; then $\sem{M}_{@}$ is strongly normalising.
\end{lemma}
\begin{proof}
  We use the same technique as used in the proof of \cite[Thm.~3.4(2)]{BerarducciB92}, and proceed by (strong) induction on the length $n$ of the longest reduction sequence from $M$ to its normal form.
  When $n = 0$, then we have that $M$ is a (SF-)normal form and thus $\sem{M}_{@}$ is strongly normalising w.r.t $\casfreduction^{\ast}$ by \Cref{lem:SF-NFImagesSN}.
  When $n > 0$ then $M$ must contain at least one redex $O\,X\,Y\,Z$. Consider an \emph{innermost} redex, whose contractum is the term $R$. Thus $M = \mathcal{C}[O\,X\,Y\,Z] \sfreduction \mathcal{C}[R] = N$ (for some (SF-)term context $\mathcal{C}$) and $X$, $Y$ and $Z$ are normal forms (since the redex is innermost).
  Now, since $M$ is strongly normalising so too is $N$, and the length of its longest reduction sequence must be strictly less than $n$ (otherwise $n$ would not be maximum). Thus by the inductive hypothesis $\sem{N}_{@}$ is strongly normalising.
  Consider now the structure of $\sem{M}_{@}$: we have $\sem{M}_{@} = \mathcal{C'}[\sem{O\,X\,Y\,Z}_{@}]$ for some ($\CurryAppSF$-)term context $\mathcal{C'}$.
  Since $X$, $Y$ and $Z$ are normal forms, by \Cref{lem:PerpetualReductions} there is a perpetual reduction sequence from $\sem{M}_{@} = \mathcal{C'}[\sem{O\,X\,Y\,Z}_{@}]$ to $\sem{N}_{@} = \mathcal{C'}[\sem{R}_{@}]$, i.e.~one which preserves divergence. Therefore, since $\sem{N}_{@}$ is strongly normalising so too is $\sem{M}_{@}$ (if it were not, neither would $\sem{N}_{@}$ be).
\end{proof}

Using Berarducci and {\Bohm}'s construction, outlined in \Cref{sec:NormalSolutions}, we obtain an encoding of SF-calculus in the $\lambda$-calculus.
\begin{definition}[Translation of SF-calculus to $\lambda$-calculus]
\label{def:LCEncoding}
  Fix a canonical representation $\phi_{\textsf{SF}}$ of $\sigSF$ w.r.t.~the rewrite rules of $\CurryAppSF$.  The mapping $\sem{\cdot}_{\lambda} = (\cdot)^{\phi_{\textsf{SF}}} \circ \sem{\cdot}_{@}$ translates SF-calculus to $\lambda$-calculus.
\end{definition}
\noindent
We leave it as an exercise to the reader to compute such a canonical representation $\phi_{\textsf{SF}}$.

We now present our main result: that $\sem{\cdot}_{\lambda}$ is a faithful encoding of SF-calculus in $\lambda$-calculus.
\begin{theorem}[Faithful Encoding of SF-calculus in $\lambda$-calculus]
\label{thm:SF2LC}
  The translation $\sem{\cdot}_{\lambda}$ of SF-calculus into $\lambda$-calculus preserves reduction and strong normalisation.
\end{theorem}
\begin{proof}
  The result follows directly from the fact that the two translations that are composed to obtain $\sem{\cdot}_{\lambda}$, namely $\sem{\cdot}_{@}$ and $(\cdot)^{\phi_{\textsf{SF}}}$, each satisfy both these properties. In the case of the former, we refer to \Cref{lem:Preservation:Reduction,lem:Preservation:StrongNormalisation}; for the latter to the results of Berarducci and \Bohm, cf. \Cref{thm:Interpretation} (note that all terms are closed, since we do not consider SF-calculus with variables).
\end{proof}

\section{Expressiveness of Factorisation: Discussion \& Related Work}
\label{sec:Discussion}

We now turn out attention to the question of the expressiveness of SF-calculus relative to $\lambda$-calculus and Combinatory Logic. On one hand our results, i.e.~\Cref{thm:SF2LC}, along with those of Jay and Given-Wilson \cite{Given-WilsonJ11}, show that SF-calculus and $\lambda$-calculus simulate the same executions. On the other, SF-calculus is structure complete whereas $\lambda$-calculus is not. Is this a contradiction and, if so, how may it be resolved?
Notwithstanding the long tradition of using simulations to characterise computational equivalence, it has since been realised that a refinement of this notion is necessary to draw richer, more meaningful comparisons. 
While a complete and in-depth analysis is not within the scope of the current paper, by discussing our results with reference to some of this work we aim to draw some concrete conclusions about the how the expressiveness of structure completeness and SF-calculus may be characterised.

\paragraph{The `Standard' Notion of Equivalence.}
From the earliest research into computability, two aspects of abstract notions of computation were identified as relevant to the idea of expressiveness. Firstly, one wants to compare the respective set of \emph{functions} that each model computes. For example, in recursion theory it was shown that the set of primitive recursive functions (on natural numbers) is a strict subset of the recursive functions (see e.g.~\cite{Tourlakis84}). At the same time, there is a requirement to compare models that operate, at a fundamental level, in diverse domains. Turing Machines, $\lambda$-calculus and Combinatory Logic are, operationally, quite different ways to compute. It was shown however that via particular (and now canonical) representations of numbers, each model can `compute' the same set of functions on natural numbers, namely the partial recursive functions. Conversely, G\"{o}delization allows each of the computations in these models to be \emph{simulated} by a partial recursive function \cite{Kleene36}. 

This idea of simulation extends to the operation of the models themselves: each model may simulate the operational behaviour of the others. Such simulations also appear to abstract away the problem of \emph{representation}: often we do want to compute functions of natural numbers, however more often we desire to compute functions over different domains; it is incumbent upon us to represent elements of the desired domain of discourse as terms of the computational model. The initial characterisation of the set of `computable' functions was over the domain of natural numbers\footnote{Turing's work is different in this respect, since he deliberately embarked on a characterisation of computable functions over a different domain, namely that of strings of arbitrary symbols.}, but what is the set of `computable' functions over some other given domain? With simulations between models it seems that one may at least lay this question aside by observing that whatever can be represented in one model may then also be represented in the other, and therefore whatever functions they do compute it is the same in both cases. A stronger conclusion would be that the set of computable functions over arbitrary domains is isomorphic to the computable functions on natural numbers.

This simulation method has long since become the standard: to show two models are of equivalent computational power, demonstrate simulations of each in the other. One model of computation is only more powerful than another, then, if it is \emph{not} possible to simulate the former in the latter. The approach has been cemented over the years, notably in Landin's now seminal work \cite{Landin66}. 
As the search space of formal computation has been explored the notion of simulation has been adapted accordingly, and there are now a number of sophisticated simulations between all sorts of models, both sequential and concurrent (e.g.~\cite{AbadiCardelli96,SangiorgiW01}). In this tradition, \Cref{thm:SF2LC} is a result showing that $\lambda$-calculus is computationally \emph{as powerful as} SF-calculus.

\paragraph{Refining the Notion of Expressiveness.}
It was already noted over two decades ago that despite the broad applicability and application of the simulation method, it is not actually a fine-grained enough notion to provide a complete and universal characterisation of expressiveness. Felleisen observed that since the languages we wish to compare are (usually) Turing-complete, other methods (than simulation) must be found in order to verify claims of relative (in-)expressiveness \cite{Felleisen91}. He proposed a framework based on the concept, from logic, of \emph{eliminability} of symbols from conservative extensions \cite{Kleene52}. One logical system $\mathcal{L}$ is a conservative extension of another system $\mathcal{L}'$ if the expressions (formulae) and theorems of the latter are subsets of those of the former. A symbol of $\mathcal{L}$ (which is not in $\mathcal{L}'$) is eliminable if there is a homomorphism (i.e.~a map preserving the syntactic structure) $\varphi : \mathsf{Exp}(\mathcal{L}) \rightarrow \mathsf{Exp}(\mathcal{L}')$ from the expressions of $\mathcal{L}$ to the expressions of $\mathcal{L}'$, which acts as identity for expressions of $\mathcal{L}'$, such that an expression $t$ is a theorem of $\mathcal{L}$ if and only if $\varphi(t)$ is a theorem of $\mathcal{L}'$. Felleisen extends this to programming languages by analogy - formulae are (syntactically valid) programs and the theorems are the terminating programs. Then we may say that language $\mathcal{L}$ is \emph{more} expressive than language $\mathcal{L}'$ when the former adds some \emph{non-eliminable} syntactic construct, i.e.~one which cannot be `translated away'. Thus the standard simulation approach is refined by imposing an extra criterion when one language is a superset of another: can the larger one be built from the smaller one using \emph{macros}?

To place SF-calculus in this framework we can consider SKF-calculus, i.e.~the extension of SF-calculus by including an additional atomic term: the familiar $\comb{K}$ combinator. Since SKF-calculus is a proper extension of Combinatory Logic, obtained by adding the $\comb{F}$ combinator, we can apply the expressiveness test of Felleisen. That is, we ask is $\comb{F}$ eliminable? The answer to this question is \emph{no}; indeed this is guaranteed by \Cref{thm:NoFactorisationInCL}. In this sense, Jay and Given-Wilson's results \emph{do} justify the claim that SF-calculus is more expressive than Combinatory Logic and $\lambda$-calculus.

\paragraph{More Abstract Notions of Computational Equivalence.}
The simulation method, and its refinement described above, are still firmly grounded in an \emph{operational} view of computation, but recent work has sought to anchor formal comparisons of expressiveness in a more general, \emph{abstract} basis. A notable contribution to this effort is the work of Boker and Dershowitz \cite{BokerD09}. They abstract the notion of computational model as simply its \emph{extension}, i.e.~the set of functions over its inherent domain that it computes, and consider simulations (encodings) between them. They derive the remarkable result that combining the standard simulation approach with the natural containment of one extensionality within another leads to a paradox: some computational models can simulate models which are \emph{strictly} more powerful in the sense that they have larger extensionalities (i.e.~compute more functions). Thus, some representations \emph{add more computational power}. 

This result begs the question: do we consider simulations via such `active' mappings to constitute an equivalence? One may suspect that the problem lies in allowing \emph{injective} mappings between domains, and that imposing stricter conditions (e.g.~bijections) would ensure `passiveness'. This is indeed the case, but adopting such restrictions is useless for most comparisons since the passive encodings are ones which are ``almost identity''. We have no choice but to allow such encodings, although we do have the option of considering a hierarchy of equivalences based on the properties of the encodings used. Boker and Dershowitz define four increasingly stricter notions of (in-)equivalence based on, respectively: injective encodings (power equivalence), corresponding to the standard approach; injective encodings for which the images are computable in the simulating models (decent power equivalence); bijective encodings (bijective power equivalence); and bijections that are inverses of each other (isomorphism). Boker and Dershowitz also show that there are models (including Turing Machines and the recursive numeric functions) which cannot simulate any stronger models; they are \emph{interpretation complete}.

We may also gain insight into the relative expressiveness of SF-calculus and $\lambda$-calculus using this framework. Our result shows that they are \emph{power equivalent}. \Cref{thm:NoFactorisationInCL} shows that they cannot be bijectively power equivalent (nor, therefore, isomorphic) since that would imply that $\lambda$-calculus could distinguish arbitrary normal forms. Thus, in this sense they are not equivalent. We do not know if SF-calculus is \emph{bijectively stronger} than $\lambda$-calculus; this would require demonstrating a bijective encoding of the latter in the former. Also, we do not know if $\lambda$-calculus is decently power equivalent to SF-calculus; we consider it at least a possibility. We would also expect that, due to its intensional capabilities, SF-calculus is interpretation complete.

Related to the work of Boker and Dershowitz, is that of Cockett and Hofstra \cite{CockettH10}, and Longley \cite{Longley14} which are both concerned with category-theoretic descriptions of abstract computational models. In these frameworks model equivalence is interpreted by categorical isomorphism, and so akin to the strongest notion of equivalence considered by Boker and Dershowitz.

\section{Conclusions \& Future Work}
\label{sec:Conclusions}

In this paper, we have considered the relationship of the recently introduced SF-calculus to the `canonical' computational model of $\lambda$-calculus and, so by extension, Combinatory Logic. We have demonstrated that SF-calculus can be faithfully encoded (i.e.~simulated) in the $\lambda$-calculus by defining a behaviourally equivalent applicative term rewriting system and then interpreting this system in $\lambda$-calculus using a construction of Berarducci and \Bohm. This result shows that SF-calculus and $\lambda$-calculus are of equivalent computational power, according to the classical interpretation of computational equivalence. We have also considered the relationship of SF-calculus to the $\lambda$-calculus using a more nuanced interpretation of equivalence, informed by research in the literature. Moreover, we hope to have exposed both SF-calculus and Berarducci and {\Bohm}'s encoding to greater prominence. We feel that they are both subjects of great interest which deserve to be better known.

With respect to future work, there is still great scope for investigating the expressiveness of SF-calculus. There are the open questions we have highlighted regarding SF-calculus as it relates to the framework of Boker and Dershowitz. The categorical and denotational natures of SF-calculus also deserve exploration. Beyond this, the wider question of how best to qualify computational expressiveness still remains; we believe the study of SF-calculus can provide further insights on this. For example, the structural completeness property is already a new metric; Jay and Vergara have also considered a strengthening of the notion of decent power equivalence in which the simulation of each model in itself via the composition of the two encodings is computable in each model \cite{JayV14}. Quantitative questions regarding expressiveness also present themselves: e.g.~our encoding of SF-calculus in $\lambda$-calculus leads to a large increase in the size of terms; are there lower bounds on the size of such increases which meaningfully quantify expressive power?

\paragraph{Acknowledgements} We would like to thank Barry Jay, and also the anonymous reviewers for helpful comments in the preparation of the final version of this paper. This research was supported by EPSRC Grant EP/K040049/1.

\bibliographystyle{eptcs}
\bibliography{bib}

\end{document}